\documentclass[12pt]{article}
\usepackage{times}
\usepackage{fullpage}
\usepackage{amsthm}
\usepackage{amsmath}
\usepackage{amsfonts}
\usepackage{amssymb}
\usepackage{dsfont}
\usepackage{color}
\usepackage{setspace}
\usepackage{url}
\usepackage{caption}
\usepackage{comment}

\newcommand{\calC }{\text{$\mathcal{C}$}}

\newcommand{\calCj }{\text{$\mathcal{C}^j$}}
\newcommand{\calCbar }{\text{$\overline{\mathcal{C}}$}}
\newcommand{\calCbarJ }{\text{$\overline{\mathcal{C}}^{\mathcal{J}}$}}
\newcommand{\calCJ }{\text{$\mathcal{C}^\mathcal{J}$}}

\newcommand{\calP }{\mbox{$\mathcal{P}$}}
\newcommand{\calJ }{\text{$\mathcal{J}$}}
\newcommand{\calN }{\text{$\mathcal{N}$}}

\newcommand{\ignore}[1]{}

\newtheorem{definition}{Definition}
\newtheorem{lemma}[definition]{Lemma}

\newtheorem{theorem}[definition]{Theorem}

\begin{document}

\renewcommand{\textfraction}{0}
\title{Local Optimality Certificates for LP Decoding of Tanner Codes}

\author{
      Nissim Halabi \thanks{School of Electrical Engineering, Tel-Aviv University, Tel-Aviv 69978, Israel. \mbox{{E-mail}:\ {\tt nissimh@eng.tau.ac.il}.}}
      \and
      Guy Even \thanks{School of Electrical Engineering, Tel-Aviv University, Tel-Aviv 69978, Israel.  \mbox{{E-mail}:\ {\tt guy@eng.tau.ac.il}.}}
}

\date{}

 \maketitle

\begin{abstract}
We present a new combinatorial characterization for local optimality
  of a codeword in an irregular Tanner code.  The main
  novelty in this characterization is that it is based on a linear
  combination of subtrees in the computation trees. These subtrees may
  have any degree in the local code nodes and may have any height (even greater than the girth).
  We expect this new characterization to lead to improvements in bounds for successful
  decoding.

  We prove that local optimality in this new characterization implies
  ML-optimality and LP-optimality, as one would expect.  Finally, we
  show that is possible to compute efficiently a certificate for the
  local optimality of a codeword given an LLR vector.
\end{abstract}

\section{Introduction} \label{sec:intro}

Modern coding theory deals with finding good codes that have efficient
decoders (see e.g.~\cite{RU08}). Many of the decoders for modern codes are sub-optimal in the
sense that they may fail to correct errors that are corrected by maximum likelihood (ML)
decoder, but their simplicity and speed make them attractive in
practice. Message-passing decoding algorithms based on
belief-propagation and linear-programming (LP) decoding are examples for
such sub-optimal decoders.

Tanner~\cite{Tan81} introduced graph representations of linear codes.
In the standard setting, check nodes compute the parity function. In
the generalized setting, check nodes use a local error-correcting
code.  One may view a check node with a local code as a coalescing of
multiple parity check nodes. Therefore, a code may have a sparser and
smaller representation when represented as a Tanner code in the generalized setting.
An example of a Tanner code with a simple bit-flipping
decoding algorithm was presented by Sipser and Spielman~\cite{SS96}.

Linear programming (LP) decoding was introduced by Feldman, Wainwright and Karger~\cite{Fel03,FWK05} for binary linear codes. LP decoding has been applied to several families of codes, among them RA codes,
turbo-like codes, LDPC codes, and expander codes. This work is motivated by the problem of
analyzing the probability of successful decoding using LP decoding for
Tanner codes. There are very few works on this problem, and they deal only with
specific cases. For example, Feldman and Stein~\cite{FS05} analyzed
special expander codes, and Goldenberg and Burshtein~\cite{GB10} deal
with repeat-accumulate codes.

The combinatorial characterization of a decoding success is based on a
test criterion that certifies the optimality of a codeword. That is,
given a received word $y$ and a codeword $x$, we consider a test that
answers the questions: is $x$ optimal with respect to $y$? and is it
unique?  We call these tests \emph{certificates} for the optimality of
a codeword.  Bounds on the word error probability may be computed by
analyzing the events for which certificates are provided.

Wiberg~\cite{Wib96} studied representations of codes using factor
graphs.  He used these representations to analyze message passing
decoding algorithms.  The analysis uses minimal combinatorial
structures (i.e., skinny trees) to characterize decoding errors when
using message passing decoding algorithms.

Koetter and Vontobel~\cite{KV06} analyzed LP decoding of regular LDPC
codes. Their analysis is based on decomposing each codeword (and
pseudocodeword) into a sum of skinny trees with uniform vertex
weights.  Arora \emph{et al.}~\cite{ADS09} extended the work in~\cite{KV06}
by introducing nonuniform weights to the vertices in the skinny trees.
For a BSC, Arora \emph{et al.} proved that local optimality implies both
ML-optimality and LP-optimality. They used analysis techniques,
similar to those used in density evolution analysis, to improve
bounds on the probability of a decoding error. This work was further extended in~\cite{HE11} to memoryless channels. The analysis in~\cite{KV06,ADS09,HE11} is limited to skinny trees, the height of which
is bounded by a quarter of the girth of the Tanner graph.

Vontobel~\cite{Von10} extended the decomposition of a codeword (and
pseudocodeword) to subtrees of the computation tree. This enabled him
to avoid the limitation of the height being bounded by the girth. The
decomposition is obtained by a random walk, and applies to irregular
Tanner graphs.

Jian and Pfister~\cite{Jian} analyzed a weighted min-sum decoding
algorithm for regular LDPC codes. They used skinny trees in the
computation tree, the height of which is greater than the girth of the
Tanner graph.  They also used local optimality to connect successful
decoding to LP-decoding.

\paragraph{Contributions.}
We present a new combinatorial characterization of local optimality for
irregular Tanner codes. This characterization uses
subtrees in the computation tree in which the degree of local code
nodes is not limited to $2$ (as opposed to skinny trees in previous
analyses). Since such trees are bigger, it is likely that this
characterization will lead to improved bounds for successful decoding.
We prove that local optimality in this characterization implies
ML-optimality and LP-optimality, as one would expect. Finally, we
show that is possible to compute efficiently a certificate for the
local optimality of a codeword given an LLR vector.

\section{Preliminaries} \label{sec:prelim}

\paragraph{Tanner-codes and Tanner graph representation.}

Let $G=(\mathcal{V} \cup \mathcal{J}, E)$ denote an edge-labeled bipartite-graph
between a set of $N$ vertices $\mathcal{V} = \{v_1,\ldots,v_N\}$ called \emph{variable nodes}, and a set of $J$ vertices $\mathcal{J} = \{C_1,\ldots,C_J\}$ called \emph{local-code nodes} where $deg_G(C_j) = n_j$.

Let $\calCbarJ \triangleq \big\{\calCbar^j :\ \calCbar^j \mathrm{\ is\
  an\ } [n_j,k_j,d_j] \mathrm{\ code } ,\ j \in [J] \big\}$ denote a
set of $J$ \emph{local-codes}. We associate every local-code
$\calCbar^j \in \calCbarJ$ with the respective local-code node $C_j
\in \calJ$. The set $E$ consists of edges $(v_i,C_j)$ such that
variable $v_i$ participates in local-code $\calCbar^j$. The labels
$\{1,\ldots,n_j\}$ of the edges incident to local-code node $C_j$
indicate the order of variable bit nodes in the corresponding
local-code $\calCbar^j$. Let $d^* \triangleq \min_{1\leqslant j
  \leqslant J}d_j$ denote the smallest minimum distance among the local codes.

Let a word $x = (x_1,\ldots,x_N) \in \mathds{F}_2^N$ denote an
assignment to variable nodes in $\mathcal{V}$. Let $\mathcal{V}_j$ denote the ordered set of variable nodes in $\calN_G(C_j)$ according to labels
of edges incident to $C_j$. Denote by $x_{\mathcal{V}_j} \in
\mathbb{F}_2^{n_j}$ the projection of the word $x = (x_1,\ldots,x_N)$
onto entries associated with $\mathcal{V}_j$.

The \emph{Tanner code} $\calC(G,\calCbarJ)$ based on labeled \emph{Tanner graph} $G$ is the code of block length $N$ with codewords $x \in \mathds{F}_2^N$ such that
$x_{\mathcal{V}_j}$ is a codeword in $\calCbar^j$ for every $j \in [J]$.

Consider a Tanner code $\calC(G,\calCbarJ)$, where
$\calCbarJ = \{\calCbar^j\}_{j \in [J]}$. We say that a word $x =
(x_1,...,x_N)$ \emph{satisfies} local-code $\calCbar^j$ if
$x_{\mathcal{V}_j} \in \calCbar^j$. Denote by \calCj\ the set of words
$x$ that \emph{satisfy} the local-code $\calCbar^j$, i.e., $\calCj
= \{x \in \mathds{F}_2^N : x_{\mathcal{V}_j} \in \calCbar^j \}$. The
resulting code $\calCj$ is the \emph{extension} of the local-code
$\calCbar^j$ from length $n_j$ to length $N$. We denote the set of
extended local-codes in $\calCbarJ$ by \calCJ. Clearly,
$\calC(G,\calCbarJ) \subseteq \calCj$. It holds that
\begin{equation}
\calC(G,\calCbarJ) = \bigcap_{j \in [J]}{\calCj}.
\end{equation}


\paragraph{LP decoding of Tanner codes.}
When transmitting over a discrete memoryless channel, the
receiver observes a measurement $y_i$ for every transmitted symbol $x_i$. In memoryless binary-input output-symmetric (MBIOS) channels, the \emph{log-likelihood ratio} (LLR) vector $\lambda \in \mathds{R}^N$ is defined by $\lambda_i (y_i) \triangleq
\ln\big(\frac{\mathbb{P}(y_i/x_i=0)}{\mathbb{P}(y_i/x_i=1)}\big)$ for every
input bit $i$. For a linear code \calC, \emph{Maximum-Likelihood (ML)
decoding} is equivalent to
\begin{equation} \label{eqn:MLdecoding}
 \hat{x}^{ML}(y) = \arg \min_{x \in \mathrm{conv}(\mathcal{C})} \langle
\lambda(y) , x \rangle,
\end{equation}
where $\mathrm{conv}(\mathcal{C})$ denotes the convex hull of the set $\mathcal{C}$.

Solving in general the optimization problem in (\ref{eqn:MLdecoding}) for linear codes is intractable. Feldman \emph{et al.}~\cite{Fel03,FWK05} introduced a linear programming relaxation for the problem of ML decoding of Tanner codes whose local codes are parity codes. LP decoding is based on minimizing an objective function over a fundamental polytope defined by a Tanner graph $G$. A natural extension of LP decoding to Tanner codes in the generalized setting is obtained by optimizing the objective function over a generalized fundamental polytope. Consider a Tanner code $\calC = \calC(G,\calCbarJ)$. The
\emph{generalized fundamental polytope} $\calP \triangleq
\calP(G,\calCbarJ)$ is defined as the convex hull
\begin{equation}
\calP \triangleq \bigcap_{\calCj \in \calCJ}{\mathrm{conv}(\calCj)}.
\end{equation}
Clearly, the generalized fundamental polytope $\calP(G,\calCbarJ)$
is a function of the (edge labeled) Tanner graph $G$ and the set of
local-codes $\calCbarJ$. Note that the representation of Tanner codes via Tanner graph and local codes is not unique.
Different representations $(G,\calCbarJ)$ of the same Tanner code \calC\ yield different generalized fundamental polytopes \calP\ for the same code \calC.
We note that for Tanner codes whose Tanner graphs have constant bounded right degree and a linear number of edges, the generalized fundamental polytope has an efficient representation. This family of codes is typically called generalized low-density parity-check codes.

Given an LLR vector $\lambda$ for a received word $y$, LP-decoding consists of solving the following optimization problem
\begin{equation} \label{eqn:LPdecoding}
\hat{x}^{LP}(y) \triangleq \arg\min_{x \in \mathcal{P}(G,\calCbarJ)} \langle
\lambda(y) , x \rangle.
\end{equation}

The difference between ML-decoding and LP-decoding is that the
fundamental polytope $\calP(G,\calCbarJ)$ may strictly contain the
convex hull of $\calC$. Vertices of $\calP(G,\calCbarJ)$ that are not
codewords of $\calC$ must have fractional components and are called \emph{pseudocodewords}.

We now introduce some graph terminology. Let $\mathcal{N}_G(v)$ denote
the set of neighbors of node $v$ in graph $G$, and for a set $S
\subseteq V$ let $\mathcal{N}_G(S)\triangleq\bigcup_{v \in
  S}\mathcal{N}_G(v)$. Let $P_{vu}(G)$ denote a shortest path between
nodes $v$ and $u$ in $G$.  Let $d_G(r,v)$ denote the
distance\footnote{Length of a shortest path} between nodes $r$ and $v$
in $G$.

An \emph{induced subgraph} is a subgraph obtained by deleting a set of
vertices. The \emph{subgraph of $G=(V,E)$ induced by $S \subseteq V$},
denoted by $G_S$, consists of $S$ and all edges in $E$, both endpoints of which
are contained in $S$. For a codeword $x\in \mathcal{C}(G) \subset
\{0,1\}^N$, let $G_x$ denote the subgraph of the Tanner graph $G$
induced by $V_x\cup\mathcal{N}(V_x)$ where $V_x = \{v_i\ |\ x_i =
1\}$.

\section{A Combinatorial Certificate for an ML Codeword} \label{sec:MLcertificate}

In this section we present combinatorial certificate, that applies both to ML-decoding and LP-decoding, for codewords of Tanner codes. A certificate is a proof that a given codeword is the unique solution of maximum-likelihood decoding and linear-programming decoding. The certificate is based on combinatorially structured weighted local configurations in the Tanner graph. These local configurations generalize the minimal configurations (skinny trees) presented by Vontobel~\cite{Von10} as extension to Arora {et al.}~\cite{ADS09}. We note that for Tanner codes, the support of each weighted local configuration is not necessarily a local valid configuration. For a given codeword, the certificate is computed by a message-passing algorithm on the Tanner graph of the code.

\emph{Notation:} Let $y \in \mathds{R}^n$ denote the word received from the channel.
Let $\lambda = \lambda(y)$ denote the LLR vector for $y$. Let $G=(\mathcal{V}\cup\mathcal{J},E)$ denote a Tanner graph, and let $\mathcal{C}(G)$ denote a Tanner code based on $G$ with local minimal distance $d^*$. Let $x \in
\mathcal{C}(G)$ be a candidate for $\hat{x}^{ML}(y)$ and
$\hat{x}^{LP}(y)$.

\begin{definition}[Path-Prefix Tree]
Consider a graph $G=(V,E)$ and a node $r \in V$. Let $\hat{V}$ denote the set of all backtrackless paths in $G$ with length at most $h$ that start at node $r$, and let
\begin{equation*}
\hat{E} \triangleq \big\{(p_1,p_2)\in\hat{V}\times\hat{V}\ | \ p_1\ \mathrm{is\ a\ prefix\ of\ p_2, \ } |p_1|+1=|p_2| \big\}.
\end{equation*}
We identify the empty path in $\hat{V}$ with $r$.
Denote by $\mathcal{T}_r^{h}(G) \triangleq (\hat{V},\hat{E})$ the \emph{path-prefix tree} of $G$ rooted at node $r$ with height $h$. We denote the fact that a path $\hat{p} \in \hat{V}$ ends at $v\in V$, by $\hat{p}\sim v$.
\end{definition}

\ignore{
\begin{definition}[Path-Prefix Tree]
Consider a graph $G=(V,E)$ and a node $r \in V$. The \emph{path-prefix tree}, $\mathcal{T}_r^{h}(G)$, of $G$ rooted at node $r$ with height $h$ is constructed as follows. Set $r$ to be the root node and iterate the following procedure $h$ times: (i) find all leaves of the tree (starting with the root), (ii) for each leaf, add as children replications of all nodes in $G$ corresponding to all neighbors except the parent node.
\end{definition}
}

The path-prefix tree is constructed by recursively unwrapping graph $G$ from a root node $v$ for $h$ iterations. When dealing with the analysis of belief propagation algorithms on graphical models, the path-prefix tree of a Tanner graph $G$ rooted at a variable node is usually referred to as the \emph{computation tree}. We make the distinction between the computation tree and the path-prefix tree since we consider also path-prefix trees of subgraphs of a Tanner graph $G$ and are not necessarily rooted at a variable node. We denote vertices in the path-prefix tree by $\hat{v}$,$\hat{u}$, etc. Vertices in $G$ are denoted by $v,u,$ etc.

The following definitions expands the combinatorial notion of minimal valid deviations~\cite{Wib96} and weighted minimal local-deviations (skinny trees)~\cite{ADS09, Von10} to the case of Tanner codes.
\begin{definition}[$d$-tree]
Consider a Tanner graph $G=(\mathcal{V}\cup\mathcal{J},E)$. A \emph{$d$-tree}, $\mathcal{T}[r,h,d](G)$, of height $h$ rooted at node $r$ is a subtree of $\mathcal{T}_r^{h}(G)$ such that every variable node has full degree and every local-code node has degree $d$.
\end{definition}

\begin{definition} [$\omega$-weighted subtree]
Consider a Tanner graph $G=(\mathcal{V}\cup\mathcal{J},E)$. Let $\mathcal{T}_{\hat{r}} = (\hat{\mathcal{V}}\cup\hat{\mathcal{J}},\hat{E})$ denote a subtree of $\mathcal{T}_r^{h}(G)$, and let $\omega=(\omega_1,\ldots,\omega_h)\in\mathds{R}_+^{h}$ denote a non-negative weight vector. Let $\mathcal{T}_{\hat{r}}^{(\omega)}:\hat{\mathcal{V}}\backslash\{\hat{r}\}\rightarrow \mathds{R}$ denote a weight function for variable nodes in $\mathcal{T}_{\hat{r}}$ as follows.
\begin{equation}
   \mathcal{T}_{\hat{r}}^{(\omega)}(\hat{v}) \triangleq \frac{\omega_t}{\deg_G(v)}\cdot \prod_{\hat{u}\in P_{\hat{r},\hat{v}}\backslash\{\hat{r},\hat{v}\}}\frac{1}{\deg_{\mathcal{T}_{\hat{r}}}(\hat{u})-1},
\end{equation}
where $t=\lceil\frac{d(\hat{r},\hat{v})}{2}\rceil$ and $\hat{v} \sim v$. Let  $\mathcal{T}_{\hat{r}}^{(\omega)}$ also denote the $\omega$-weighted subtree $\mathcal{T}_{\hat{r}}$ rooted at node $\hat{r}$.
\end{definition}

For any $\omega$-weighted subtree $\mathcal{T}_{\hat{r}}^{(\omega)}$ of $\mathcal{T}_{r}^{h}(G)$, let $\pi_G[\mathcal{T}_{\hat{r}}^{(\omega)}] \in \mathds{R}^{|\mathcal{V}|}$ denote the projection of $\mathcal{T}_{\hat{r}}^{(\omega)}$ to the Tanner graph $G$. That is, for every variable node $v$ in $G$,
\begin{equation}
\pi_G[\mathcal{T}_{\hat{r}}^{(\omega)}](v) = \begin{cases} \sum_{\hat{v}:\hat{v}\sim v}\mathcal{T}_{\hat{r}}^{(\omega)}(\hat{v}) & \mathrm{if\ } \{\hat{v}:\hat{v}\sim v\}\neq\emptyset,\\ 0 & \mathrm{otherwise.} \end{cases}
\end{equation}

For two vectors $x \in \{0,1\}^N$ and $f \in [0,1]^N$, let $x\oplus f \in [0,1]^N$ denote the \emph{relative point} defined by $(x\oplus f)_i = |x_i-f_i|$~\cite{Fel03}. The following definition is an extension of local-optimality~\cite{ADS09,Von10} to Tanner codes on memoryless channels.

\begin{definition}[local-optimality] \label{def:localOptimality}
Let $\mathcal{C}(G) \subset \{0,1\}^N$ denote a Tanner code with minimal local-distance $d^*$, and let $\omega \in [0,1]^h\backslash\{0^N\}$ denote a non-negative weight vector of length $h$. For any integer $2 \leqslant i \leqslant d^*$, let $\mathcal{B}_i^{(\omega)}$ denote the set of all vectors corresponding to projections by $\omega$-weighted $i$-trees to $G$, i.e., $\mathcal{B}_i^{(\omega)} = \big\{\pi_G[\mathcal{T}^{(\omega)}[r,2h,i](G)]\ \big|\  r \mathrm{\ is\ a\ variable\ node\ in\ }G\big\}$.
A codeword $x \in \{0,1\}^N$ is \emph{$(h,\omega,i)$-locally optimal for $\lambda \in \mathds{R}^N$} if for all vectors $\beta \in \mathcal{B}_i^{(\omega)}$,
\begin{equation}
\langle \lambda,x \oplus \beta \rangle > \langle \lambda, x \rangle.
\end{equation}
\end{definition}

Note that $\mathcal{B}_i^{(\omega)} \subseteq [0,1]^N$  for every weight vector $\omega \in [0,1]^h$. Based on random walks on the Tanner graph, Vontobel showed that $(h,\omega,2)$-local optimality is sufficient both for global optimality and LP optimality. The random walks are defined in terms derived from the generalized fundamental polytope.
We extend the results of Vontobel~\cite{Von10} to ``thicker'' skinny-trees by using probabilistic combinatorial arguments on graphs and the properties of graph cover decoding~\cite{VK05}.
Specifically, we prove that $(h,\omega,i)$-local optimality, for any $2 \leqslant i \leqslant d^{*}$, implies LP optimality (Theorem~\ref{thm:LPsufficient}).
\ignore{We first show how to extend the
proof that $(h,\omega,i)$-local optimality implies global optimality.}
Given the decomposition of Lemma~\ref{lemma:IntegralDecomposition} proved in  Section~\ref{sec:decomposition}, the following theorem is obtained by modification of the proof of~\cite[Theorem 2]{ADS09} or ~\cite[Theorem 6]{HE11}.

\begin{theorem}[local-optimality is sufficient for ML]\label{thm:MLsufficient}
Let $\mathcal{C}(G)$ denote a Tanner code with minimal local-distance $d^*$. Let $h$ be some positive integer and $\omega=(\omega_1,\ldots,\omega_h)\in[0,1]^h$ denote a non-negative weight vector. Let $\lambda\in\mathds{R}^N$ denote the LLR vector received from the channel, and suppose that $x$ is an $(h,\omega,i)$-locally optimal codeword for $\lambda$ and some $2\leqslant i \leqslant d^{*}$. Then $x$ is also the unique maximum-likelihood codeword for $\lambda$.
\end{theorem}

  \begin{proof}
    We use the decomposition proved in Section~\ref{sec:decomposition}
    to show that for every codeword $x' \neq x$, $\langle \lambda , x'
    \rangle > \langle \lambda , x \rangle$.  Since $z \triangleq x
    \oplus x'$ is a codeword, by
    Lemma~\ref{lemma:IntegralDecomposition} there exists a
    distribution over the set $\mathcal{B}_i^{(\omega)}$, such that
    $\mathds{E}_{\beta \in \mathcal{B}_i^{(\omega)}} \beta = \alpha z$. Let $f:[0,1]^N
    \rightarrow \mathds{R}$ be the affine linear function defined by
    $f(u) \triangleq \langle \lambda , x \oplus u \rangle = \langle
    \lambda , x \rangle + \sum_{i=1}^{N}(-1)^{x_i}\lambda_i u_i$.
    Then,
    \begin{eqnarray*}
      \langle \lambda , x \rangle &<& \mathds{E}_{\beta \in \mathcal{B}_i^{(\omega)}} \langle \lambda , x \oplus \beta \rangle \ \ \ (\text{by local-optimality of $x$}) \\
      &=& \langle \lambda , x \oplus \mathds{E} \beta \rangle \ \ \ \ \ \ \ \ \ \ \ \ (\text{by linearity of $f$ and linearity of expectation}) \\
      &=& \langle \lambda , x \oplus \alpha z \rangle  \ \ \ \ \ \ \ \ \ \ \ \ \ (\text{by Lemma \ref{lemma:IntegralDecomposition}})\\
      &=& \langle \lambda , (1-\alpha)x + \alpha(x \oplus z) \rangle\\
      &=& \langle \lambda , (1-\alpha)x + \alpha x' \rangle \\
      &=& (1-\alpha) \langle \lambda , x \rangle + \alpha \langle \lambda , x' \rangle.
    \end{eqnarray*}
    which implies that $\langle \lambda , x' \rangle > \langle \lambda
    , x \rangle$ as desired.
  \end{proof}

In order to prove a sufficient condition for LP optimality, we
consider graph cover decoding introduced by Vontobel and Koetter~\cite{VK05}. We note that the characterization of graph cover decoding and its connection to LP decoding~\cite{VK05}, can be extended to the case of Tanner codes in the generalized setting.
We use the terms and notation of Vontobel and Koetter \cite{VK05} in the statement of Lemma~\ref{lemma:coverOptimality}. The following lemma shows that local-optimality based on $i$-trees is preserved after lifting to an $M$-cover. Note that the weight vector
must be scaled by the cover degree $M$.

\begin{lemma} \label{lemma:coverOptimality}
Let $\mathcal{C}(G)$ denote a Tanner code with minimal local-distance $d^*$, and let $\tilde{G}$ denote any
$M$-cover of $G$. Let $\omega \in [0,\frac{1}{M}]^h\backslash\{0^h\}$ for some positive integer $h$. Suppose that $x \in \mathcal{C}(G)$ is an
$(h,\omega,i)$-locally optimal codeword for $\lambda \in \mathds{R}^N$ for some $2\leqslant i \leqslant d^{*}$. Let
$\tilde{x}=x^{\uparrow M}\in \mathcal{C}(\tilde{G})$ and
$\tilde{\lambda}=\lambda^{\uparrow M} \in \mathds{R}^{N\cdot M}$
denote the $M$-lifts of $x$ and $\lambda$, respectively. Then
$\tilde{x}$ is an $(h,M \cdot \omega,i)$-locally optimal codeword for $\tilde{\lambda}$.
\end{lemma}

\begin{proof}
Assume that $\tilde{x}=x^{\uparrow M}$ is not a $(h,M \cdot
\omega, i)$-locally optimal codeword for $\tilde{\lambda}=\lambda^{\uparrow
M}$. Then, there exists an $i$-tree $\mathcal{T}=\mathcal{T}[\tilde{r},h,i](\tilde{G})$ rooted at some variable node $\tilde{r}\in\tilde{\mathcal{V}}$ and a weight vector $\omega$, such that the projection $\tilde{\beta} = \pi_{\tilde{G}}[\mathcal{T}^{(M \cdot
\omega)}] \in [0,1]^{N\cdot M}$ of the $(M\cdot\omega)$-weighted $i$-tree $\mathcal{T}^{(M \cdot
\omega)}$ onto $\tilde{G}$ satisfies
\begin{equation} \label{eqn:proof1}
\langle \tilde{\lambda} , \tilde{x} \oplus \tilde{\beta} \rangle \leqslant \langle \tilde{\lambda} , \tilde{x} \rangle.
\end{equation}
Note that for $\tilde{x}\in \{0,1\}^{N\cdot M}$ and its projection $x = p(\tilde{x})\in \mathds{R}^N$, it holds that
\begin{eqnarray}
   \frac{1}{M} \langle \tilde{\lambda}, \tilde{x} \rangle &=& \langle \lambda , x \rangle,  \mathrm{\ \ \ and} \label{eqn:proof2}\\
  \frac{1}{M} \langle \tilde{\lambda} , \tilde{x} \oplus \tilde{\beta} \rangle &=& \langle \lambda , x \oplus \beta \rangle, \label{eqn:proof3}
\end{eqnarray}
where $\beta = \pi_{G}[\mathcal{T}^{(\omega)}] \in [0,1]^{N}$ is the projection of the $\omega$-weighted $i$-tree $\mathcal{T}$ onto the base graph $G$.
From (\ref{eqn:proof1}),
(\ref{eqn:proof2}), and (\ref{eqn:proof3}) we get that $\langle \lambda
, x \rangle \geqslant \langle \lambda , x \oplus \beta \rangle$,
contradicting our assumption on the $(h,\omega,i)$-local optimality of $x$.
Therefore, $\tilde{x}$ is a $(h,M \cdot \omega,i)$-locally optimal codeword
for $\tilde{\lambda}$ in $\mathcal{C}(\tilde{G})$.
\end{proof}

The following theorem is obtained as a corollary of  Theorem~\ref{thm:MLsufficient} and Lemma~\ref{lemma:coverOptimality}. The proof is based on arguments utilizing properties of graph cover decoding. Those arguments are used for a reduction from ML-optimality to LP-optimality similar to the reduction presented in the proof of~\cite[Theorem 8]{HE11}.

\begin{theorem}[local optimality is sufficient for LP optimality]\label{thm:LPsufficient}
For every Tanner code $\mathcal{C}(G)$ with minimal local-distance $d^*$, there exists a constant $M$ such that, if
\begin{enumerate}
\item $\omega \in [0,\frac{1}{M}]^h\backslash\{0^h\}$,and
\item $x$ is an $(h,\omega,i)$-locally optimal codeword for $\lambda \in \mathds{R}^N$ and some $2\leqslant i \leqslant d^{*}$,
\end{enumerate}
then $x$ is also the unique optimal LP solution given $\lambda$.
\end{theorem}

\subsection{Verifying local optimality}
Let $G=(\mathcal{V}\cup\mathcal{J},E)$ denote a Tanner graph, and let
$\mathcal{C}(G)$ denote a Tanner code with minimal
local-distance $d^*$. Let $h$ denote a positive integer and $\omega
\in [0,1]^h$. Consider a codeword $x \in \mathcal{C}(G)$ and any
integer $2 \leqslant i \leqslant d^*$. Note that for a given LLR
vector $\lambda$, the weighted $i$-tree
$\mathcal{T}^{(\omega)}[r,h,i](G)$ that minimizes $\langle \lambda,
x\oplus \beta \rangle$ for all vectors $\beta$ corresponding to
projections of $\omega$-weighted $i$-trees rooted at $r$, can be
computed by a simple message passing algorithm. The messages are
propagated from the leaves of $\mathcal{T}^{(\omega)}_r(G)$ to the
root $r$. In fact, message-passing algorithms on computation trees run
simultaneously for every root in $G$. After $h$ iterations we can
verify if the codeword $x$ is indeed $(h,\omega,i)$-locally optimal
for $\lambda$ (according to Definition~\ref{def:localOptimality}). We
can therefore compute an $(h,\omega,i)$-local optimality certificate
in $O(|E|\cdot h)$ time.

\ignore{
Describe an iterative algorithm: weighted min-sum algorithm.

\begin{theorem}
Let $\mathcal{C}(G)$ denote a Tanner code whose local-codes' minimum distance is at least $d$. Let $\omega \in \mathds{R}^h$ denote a non-negative weight vector. Let $\lambda\in\mathds{R}^N$ denote the LLR vector received from the channel, and suppose that $x$ is an $(h,\omega,i)$-locally optimal codeword for $\lambda$ and some $2\leqslant i \leqslant d$. Then, the $\omega$-reweighted min-sum algorithm on input $\lambda$ and $i$ computes $x$ in $h$ iterations.
\end{theorem}
}

\section{Constructing Codewords from Weighted Trees Projections} \label{sec:decomposition}

This section features Lemma~\ref{lemma:IntegralDecomposition}, which is the key structural lemma in the proof of Theorem~\ref{thm:MLsufficient}. This Lemma shows that every codeword of a Tanner code can be constructed by a summation over a finite set of projections of weighted trees in the computation trees of $G$.

\begin{lemma} \label{lemma:IntegralDecomposition}
Let $\mathcal{C}(G)$ denote a Tanner code with minimal local-distance $d^*$, and let $h$ denote some positive integer. For every codeword $x \neq 0^N$, and for every $2 \leqslant i \leqslant d^{*}$, there exists a distribution over $i$-trees $\mathcal{T}$ of $G$ of height $h$ and a positive integer $H$ such that, for every weight vector $\omega \in [0,\frac{1}{H}]^h\backslash\{0^h\}$, there exists an $\alpha \in (0,1]$, such that
  \[\mathds{E}_{\mathcal{T} \in \mathcal{B}_i^{(\omega)}}\big[
\pi_G[\mathcal{T}]\big] = \alpha x.\]
\end{lemma}

We first prove that every codeword $x\in\mathcal{C}(G)$ can be decomposed into exactly $\|x\|_1$ weighted path-prefix trees (see Lemma~\ref{lemma:prefixDecomposition}). Then we show that every weighted path-prefix tree can be decomposed to a set of weighted $i$-trees (see Lemma~\ref{lemma:skiniesDecomposition}). Putting these two results together yields Lemma~\ref{lemma:IntegralDecomposition}.

\begin{lemma}\label{lemma:prefixDecomposition}
Let $\mathcal{C}(G)$ denote a Tanner code and let $h$ denote some positive integer. For every codeword $x \neq 0^N$, and for every weight vector $\omega \in \mathds{R}_+^h$,
\begin{equation*}
\big(\sum_{t=1}^{h}\omega_t\big)\cdot x = \sum_{r:x_r=1}\pi_G[\mathcal{T}_r^{(\omega)}(G_x)].
\end{equation*}
\end{lemma}

\begin{proof}
Let us consider two variable nodes $u,v \in G_x$.
Notice that $|\{\hat{v} \in \mathcal{T}_u^h(G_x)\ : \ \hat{v} \sim v\}| = |\{\hat{u} \in \mathcal{T}_v^h(G_x)\ : \ \hat{u} \sim u\}|$. Indeed, for every path from the root of $\mathcal{T}_u^h(G_x)$ to a node $\hat{v} \in \{\hat{v}: \hat{v}\sim v\}$, there exists a unique reversed path in $\mathcal{T}_v^h(G_x)$ from the root to a node $\hat{u}$ such that $\hat{u}\sim u$.
Let $\overrightarrow{p}=(v,\ldots,\hat{r})$ denote a path in the path-prefix tree $\mathcal{T}_v^h$ rooted at $v$, then $\overleftarrow{p} = (r,\ldots,\hat{v})$ denotes the corresponding reversed path in the path-prefix tree $\mathcal{T}_r^h$.

Consider an all-one weight vector $\eta=1^h$. In
(\ref{eqn:recursive})-(\ref{eqn:base}), let $\mathcal{T}_r^{(\eta)}
\triangleq \mathcal{T}_r^{(\eta)}(G_x)$, $\deg(\cdot) \triangleq
\deg_{G_x}(\cdot)$, $d(\cdot,\cdot) \triangleq
d_{\mathcal{T}_v^{2h}(G_x)}(\cdot,\cdot)$, $\hat{r} \sim r$, and
$\hat{u} \sim u$. Let $q\circ p$ denote the concatenation of path $q$ with path $p$. Equation~(\ref{eqn:recursive}) holds for every $1\leqslant i \leqslant 2h$.

  \begin{eqnarray}\label{eqn:recursive}
  \nonumber
 \sum_{\{\overrightarrow{p}=(v,\ldots,\hat{r}):d(v,\hat{r})=i\}}\mathcal{T}_r^{(\eta)}(\overleftarrow{p}) &=&
\sum_{\{\overrightarrow{q}=(v,\ldots,\hat{u}):d(v,\hat{u})=i-1\}}\sum_{\{\hat{r}\in \mathcal{N}(\hat{u}): d(v,\hat{r})=i\}}\mathcal{T}_r^{(\eta)}\big(\overleftarrow{\overrightarrow{q}\circ (r)}\big)\\
  \nonumber
  &=&
  \sum_{\{\overrightarrow{q}=(v,\ldots,\hat{u}):d(v,\hat{u})=i-1\}}\sum_{\{\hat{r}\in \mathcal{N}(\hat{u}): d(v,\hat{r})=i\}}\frac{1}{\deg(u)-1}\mathcal{T}_u^{(\eta)}(\overleftarrow{q})\\
  \nonumber
  &=&
  \sum_{\{\overrightarrow{q}=(v,\ldots,\hat{u}):d(v,\hat{u})=i-1\}}\mathcal{T}_u^{(\eta)}(\overleftarrow{q})\cdot\sum_{\{\hat{r}\in \mathcal{N}(\hat{u}): d(v,\hat{r})=i\}}\frac{1}{\deg(u)-1}\\
  &=&
  \sum_{\{\overrightarrow{q}=(v,\ldots,\hat{u}):d(v,\hat{u})=i-1\}}\mathcal{T}_u^{(\eta)}(\overleftarrow{q}).
\end{eqnarray}

Note that the reversed paths $\overleftarrow{p}$ and $\overleftarrow{q}$ in the summations of (\ref{eqn:recursive}) end at a node $\hat{v}$ such that $\hat{v}\sim v$. Equation~(\ref{eqn:recursive}) implies that the sum of all $\eta$-weighted assignments to nodes $\hat{v}\sim v$ in $\{\mathcal{T}_r^{(\eta)}(G_x)\ : \ x_r=1\}$ that correspond to paths of length $i$ does not depend on $i$.

In particular, for $i=1$, $\sum_{\{\overrightarrow{p}=(v,\hat{r})\}}\mathcal{T}_r^{(\eta)}(\overleftarrow{p})=1$. It follows that for every $1\leqslant i \leqslant 2h$,
\begin{equation} \label{eqn:base}
\sum_{\{\overrightarrow{p}=(v,\ldots,\hat{r}):d(v,\hat{r})=i\}}\mathcal{T}_r^{(\eta)}(\overleftarrow{p}) = 1.
\end{equation}

Note that for every two variable nodes $v,r$, it holds that $\mathcal{T}_r^{(\omega)}(\hat{v}) = \omega_{d(r,\hat{v})/2}\cdot\mathcal{T}_r^{(\eta)}(\hat{v})$. Hence, $\sum_{\{\overrightarrow{p}=(v,\ldots,\hat{r}):d(v,\hat{r})=2i\}}\mathcal{T}_r^{(\omega)}(\overleftarrow{p}) = \omega_i$. We conclude that for every variable node $v$ in $G_x$
\begin{equation}
\sum_{r:x_r=1}\pi[\mathcal{T}_r^{(\omega)}(G_x)](v) = \big(\sum_{i=1}^{h}\omega_i\big),
\end{equation}
and the claim follows.
\end{proof}

\begin{lemma} \label{lemma:skiniesDecomposition}
For every connected subgraph $G_S$ of a Tanner graph $G$, let $d$ denote the minimal degree of a local-code node in $G_S$. Then for every variable node $r \in G_S$, a positive integer $h$, $2 \leqslant i \leqslant d$, and every weight vector $\omega \in \mathds{R}_+^h$, it holds that
  \[ \mathcal{T}_r^{(\omega)}(G_S) = \mathds{E}\big[
\mathcal{T}^{(\omega)}[r,2h,i](G_S)\big]\]
with respect to a uniform distribution over $i$-trees $\mathcal{T}$ of $G_S$ rooted at $r$ with height $2h$.
\end{lemma}

\begin{proof}
Consider a subgraph $G_S$ of a Tanner graph $G$, and a positive integer $i\leqslant d$. Let $\mathcal{T}_r^{(\omega)}(G_S)$ denote an $\omega$-weighted path-prefix tree rooted at node $r$ with height $2h$. We want to show that the uniform distribution over $\omega$-weighted $i$-trees has the property that the expectation of trees over the distribution equals $\mathcal{T}_r^{(\omega)}(G_S)$.

We grow an $i$-tree rooted at $r$ randomly in the path-prefix tree
$\mathcal{T}_r^{2h}(G_S)$. That is, start from the root $r$. For each
variable node take all it's children, and for each local-code node
choose $i$ distinct children uniformly at random. Let
$\mathcal{T}[r,2h,i]$ denote such a random $i$-tree, and consider a
variable node $\hat{v} \in \mathcal{T}_r^{2h}(G_S)$. Note that
$\mathcal{T}^{(\omega)}[r,2h,i](\hat{v})$ is constant and does not
depend on the random process. Equation~(\ref{eqn:i-tree-grow}) develops the equality
\begin{align*}
\mathds{E}\big[\mathcal{T}^{(\omega)}[r,2h,i](\hat{v})\big] &=
\mathcal{T}_r^{(\omega)}(\hat{v}).
\end{align*}
  \begin{eqnarray}\label{eqn:i-tree-grow}
\mathds{E}\big[\mathcal{T}^{(\omega)}[r,2h,i](\hat{v})\big] &=&
\sum_{\{\mathcal{T}[r,2h,i] \in\mathcal{T}_r^{2h}(G_S)\}}\mathbb{P}(\mathcal{T}[r,2h,i])\cdot \mathcal{T}^{(\omega)}[r,2h,i](\hat{v})\nonumber \\
    &=& \sum_{\{\mathcal{T}[r,2h,i] \in\mathcal{T}_r^{2h}(G_S):\hat{v}\in\mathcal{T}[r,2h,i]\}}\mathbb{P}(\mathcal{T}[r,2h,i])\cdot \mathcal{T}^{(\omega)}[r,2h,i](\hat{v})\nonumber \\
    &=&  \mathcal{T}^{(\omega)}[r,2h,i](\hat{v}) \cdot \sum_{\{\mathcal{T}[r,2h,i] \in\mathcal{T}_r^{2h}(G_S):\hat{v}\in\mathcal{T}[r,2h,i]\}}\mathbb{P}(\mathcal{T}[r,2h,i])\nonumber \\
    &=&  \mathcal{T}^{(\omega)}[r,2h,i](\hat{v}) \cdot \mathbb{P}(\hat{v}\in\mathcal{T}[r,2h,i])\nonumber \\
    &=&  \mathcal{T}^{(\omega)}[r,2h,i](\hat{v})\cdot\prod_{\hat{u}\in P_{r\hat{v}}\backslash\{r,\hat{v}\}\cap \hat{\mathcal{J}}}
    \frac{i-1}{\deg(\hat{u})-1}\nonumber \\
    &=&  \frac{\omega_{d(r,\hat{v})/2}}{\deg(\hat{v})\cdot(i-1)^{d(\hat{r},\hat{v})/2}}\cdot\prod_{\hat{u}\in P_{r\hat{v}}\backslash\{r,\hat{v}\}\cap \hat{\mathcal{V}}}
    \frac{1}{\deg(\hat{u})-1}\cdot\prod_{\hat{u}\in P_{r\hat{v}}\backslash\{r,\hat{v}\}\cap \hat{\mathcal{J}}}
    \frac{i-1}{\deg(\hat{u})-1}\nonumber \\
    &=& \frac{\omega_{d(r,\hat{v})/2}}{\deg(\hat{v})}\cdot\prod_{\hat{u}\in P_{r\hat{v}}\backslash\{r\}}
    \frac{1}{\deg(\hat{u})-1}\nonumber \\
    &=&  \mathcal{T}_r^{(\omega)}(\hat{v})
  \end{eqnarray}
as required.
\end{proof}

\section{Conclusion} \label{sec:conclusion}

A new combinatorial characterization for local optimality of a
codeword in an irregular Tanner code is presented.  The
main novelty in this characterization is that it is based on a linear
combination of subtrees in the computation trees. These subtrees may
have any degree $i$ in the local code nodes, for $2\leq i \leq d^*$.
This increased degree enables each subtree to be larger than a skinny
tree.  The larger a subtree is in the decomposition, the smaller the
probability that its cost is negative. Thus, we expect this new
characterization to lead to improvements in bounds for successful
decoding.

It is interesting to develop and analyze decoding algorithms for
irregular Tanner codes that are based on this new
characterization of local optimality.


\end{document}